\documentclass[conference]{IEEEtran}

\usepackage{amsmath,amssymb,amsthm}
\usepackage{physics}
\usepackage{graphicx} 
\usepackage{subcaption} 
\usepackage{bm} 
\usepackage{xcolor} 
\usepackage{tikz} 
\usepackage[compat=0.6]{yquant} 
\usepackage{float}
\usepackage{multirow} 
\usepackage{hyperref} 
\usepackage{cleveref} 
\usepackage[noadjust]{cite} 

\newtheorem{theorem}{Theorem}
\newtheorem{lemma}[theorem]{Lemma}
\newtheorem{definition}[theorem]{Definition}
\newtheorem{corollary}[theorem]{Corollary}

\newcommand{\Desc}[2]{\makebox[10ex][l]{\textbf{#1:}} #2\\}

\newcommand{\period}{.}

\newcommand{\BCube}[1]{\qty{0,1}^{#1}}
\newcommand{\Reg}[2]{\ket{#1}_{\mathtt{#2}}}
\newcommand{\HW}[1]{\abs{#1}}
\newcommand{\HD}[2]{\abs{#1\oplus #2}}
\newcommand{\Spl}[2]{G^{(#1,#2)}}
\newcommand{\Elm}[1]{\hat{#1}}
\newcommand{\Flip}[1]{\overline{#1}}
\newcommand{\Conj}[1]{#1^*}
\newcommand{\ceil}[1]{\left\lceil #1 \right\rceil}
\newcommand{\floor}[1]{\left\lfloor #1 \right\rfloor}

\DeclareMathOperator{\CNOT}{CNOT}
\DeclareMathOperator{\Tof}{Tof}

\newcommand{\Ord}[1]{O\qty(#1)}
\newcommand{\ord}[1]{o\qty(#1)}
\newcommand{\Tht}[1]{\Theta\qty(#1)}
\newcommand{\Omg}[1]{\Omega\qty(#1)}
\newcommand{\omg}[1]{\omega\qty(#1)}
\newcommand{\ply}[1]{\operatorname{poly}\qty(#1)}

\colorlet{DiffDel}{red!20}
\colorlet{DiffAdd}{green!30}
\usepackage{soul}
\usepackage{tcolorbox}
\tcbset{boxrule=0pt,boxsep=0pt,left=0pt,right=0pt,top=0pt,bottom=0pt,arc=0pt,colframe=white,halign=left}

\begin{document}

\floatstyle{ruled}
\newfloat{algorithm}{tbhp}{lot}
\floatname{algorithm}{Algorithm}

\title{Towards Optimal Circuit Size for Sparse Quantum State Preparation}

\author{%
  \IEEEauthorblockN{%
    Rui Mao \IEEEauthorrefmark{1}\IEEEauthorrefmark{2}\IEEEauthorrefmark{3},%
    Guojing Tian \IEEEauthorrefmark{1}\IEEEauthorrefmark{2}\IEEEauthorrefmark{4} and%
    Xiaoming Sun \IEEEauthorrefmark{1}\IEEEauthorrefmark{2}\IEEEauthorrefmark{5}}
  \IEEEauthorblockA{\IEEEauthorrefmark{1}%
    State Key Lab of Processors, Institute of Computing Technology, Chinese Academy of Sciences, Beijing 100190, China}
  \IEEEauthorblockA{\IEEEauthorrefmark{2}%
    School of Computer Science and Technology, University of Chinese Academy of Sciences, Beijing 100049, China}
  \IEEEauthorblockA{\IEEEauthorrefmark{3}%
    maorui21b@ict.ac.cn}
  \IEEEauthorblockA{\IEEEauthorrefmark{4}%
    tianguojing@ict.ac.cn}
  \IEEEauthorblockA{\IEEEauthorrefmark{5}%
    sunxiaoming@ict.ac.cn}
}

\maketitle

\begin{abstract}
 Compared to general quantum states, the sparse states arise more frequently in  the field of quantum computation.  In this work, we consider the preparation for $n$-qubit sparse quantum states  with $s$ non-zero amplitudes and propose two algorithms.  The first algorithm uses $O(ns/\log n + n)$ gates, improving upon previous  methods by $O(\log n)$.  We further establish a matching lower bound for any algorithm which is not  amplitude-aware and employs at most $\operatorname{poly}(n)$ ancillary qubits.  The second algorithm is tailored for binary strings that exhibit a short  Hamiltonian path.  An application is the preparation of $U(1)$-invariant state with $k$ down-spins  in a chain of length $n$, including Bethe states, for which our algorithm
  constructs a circuit of size $O\left(\binom{n}{k}\log n\right)$.  This surpasses previous results by $O(n/\log n)$ and is close to the lower  bound $O\left(\binom{n}{k}\right)$. Both the two algorithms shrink the existing gap theoretically and provide increasing advantages numerically.
\end{abstract}

\section{Introduction}

Quantum state preparation (QSP) is a crucial subroutine in various quantum
algorithms, such as Hamiltonian simulation
\cite{berry2015simulating,low2017optimal,low2019hamiltonian,berry2015hamiltonian}
and quantum machine learning
\cite{schuld2015introduction,biamonte2017quantum,kerenidis2016quantum,rebentrost2018quantum,harrow2009quantum,wossnig2018quantum,kerenidis2019q,kerenidis2021quantum,rebentrost2014quantum},
to load classical data into a quantum computer.
While preparing a general $n$-qubit quantum state requires $\Omg{2^n}$ 1- and
2-qubit quantum gates \cite{plesch2011quantum}, practical scenarios often
involve structured classical data, allowing for a reduction in circuit size
\cite{gleinig2021efficient,araujo2021entanglement,rattew2022preparing}.
One typical example is sparse quantum state preparation (SQSP), where the data
to be loaded is sparse.
Sparse states are not only efficient to prepare but also hold practical
significance.
Many interesting quantum states are sparse, including the GHZ states, W states
\cite{dur2000three}, Dicke states \cite{bartschi2019deterministic}, thermofield
double states \cite{cottrell2019build}, and Bethe states
\cite{van2021preparing}.
Moreover, several quantum algorithms specifically require sparse initial states, such as the quantum Byzantine agreement algorithm \cite{ben2005fast}.

We briefly summarize previous work on SQSP.
Gleinig \textit{et al.} presented an efficient algorithm for SQSP by
reducing the cardinality of basis states \cite{gleinig2021efficient}.
The Ref. \cite{malvetti2021quantum} transformed SQSP to QSP by pivoting a sparse
state into a smaller general state, which serves as a building block in their
algorithm for constructing sparse isometries.
In \cite{de2022double}, the authors proposed the CVO-QRAM algorithm, whose cost depends
on both the sparsity of the state and the Hamming weights of the bases.
Decision diagrams have been employed in \cite{mozafari2022efficient} to reduce circuit
size for states with symmetric structures.
Their algorithm also efficiently prepares sparse states.
In \cite{ramacciotti2023simple}, it was demonstrated that the Grover-Rudolph
algorithm \cite{grover2002creating} for QSP is also efficient for SQSP, and
gave an improvement based on a similar idea to
\cite{malvetti2021quantum}'s.
While their techniques differ, the aforementioned algorithms require $\Tht{sn}$
1- and 2-qubit gates to prepare an $n$-qubit states with $s$ non-zero entries
in the worst case.
On the other hand, the trivial lower bound of circuit size is $\Omg{s+n}$ with
no better lower bound currently known to the best of our knowledge, leaving an
open gap of $\Ord{n}$.
Additionally, there are efforts to reduce the circuit depth.
\cite{zhang2022quantum} achieved a circuit depth of $\Theta(\log(sn))$
for SQSP at the expense of using $\mathcal{O}(sn\log s)$ ancillary qubits.

In this work, we aim at the optimizing circuit size for SQSP.
We develop two algorithms within the framework of CVO-QRAM.
The first algorithm achieves a circuit size of $\Ord{\frac{ns}{\log n} + n}$,
improving previous upper bound by a factor of $\Ord{\log n}$.
We show that our algorithm is asymptotically optimal under certain assumptions
of SQSP algorithm.
One of the assumptions, which is satisfied by previous SQSP algorithms, is that
the algorithm must not be amplitude-aware.
That is, the structure of the output circuit depends solely on the basis states
rather than the amplitudes.
The second algorithm performs better when the basis states can be efficiently
iterated, achieving a circuit size of $\Ord{L \log n}$, where $L$ is the number
of bits flipped during iteration.
This algorithm can be applied to prepare $U(1)$-invariant states, surpassing
previous work by a factor of $\Ord{\frac{n}{\log n}}$.
The comparison of our algorithms with previous ones is summarized in \mbox{\cref{tab:430v}}.

\begin{table}[htbp]
  \centering
  \caption{Comparison of circuit complexity for sparse state preparation and $U(1)$-invariant state preparation}
  \label{tab:430v}
  \begin{tabular}{|c|c|c|c|}
    \hline
    Task                    & Algorithm                     & Circuit Size                & \#Ancilla  \\
    \hline
    \multirow{6}{*}{Sparse} & \cite{gleinig2021efficient}   & $O(ns)$                     & 0          \\
                            & \cite{malvetti2021quantum}    & $O(ns)$                     & 0          \\
                            & \cite{de2022double}           & $O(ns)$                     & 0          \\
                            & \cite{mozafari2022efficient}  & $O(ns)$                     & 1          \\
                            & \cite{ramacciotti2023simple}  & $O(ns)$                     & 0          \\
                            & \textbf{BE-QRAM}              & $\Ord{\frac{ns}{\log n}}$  & 2 \\
    \hline
    \multirow{2}{*}{$U(1)$} & \cite{raveh2024deterministic} & $\Ord{\binom{n}{k} n}$      & 0          \\
                            & \textbf{LT-QRAM}              & $\Ord{\binom{n}{k} \log n}$ & $O(n)$     \\
    \hline
  \end{tabular}
\end{table}

The remainder of the paper is structured as follows.
In \cref{sec:6zw4}, we define terminology and give the problem formulation.
In \cref{sec:8aj1}, we review the CVO-QRAM algorithm.
In \cref{sec:yakv}, we introduce our first algorithm and proves its optimality
under assumptions.
In \cref{sec:ssul}, we present our second algorithm and demonstrate its
application to $U(1)$-invariant states.
We conclude and discuss our results in \cref{sec:u6jp}.

\section{Preliminary}\label{sec:6zw4}
\paragraph{Notations}
Denote $[n] = \qty{1, 2, \dots, n}$.
All logarithms $\log (\cdot)$ are base 2.
For binary strings $x, y \in \BCube{n}$: $x_i$ is the $i$th bit of $x$;
$\HW{x}$ means the Hamming weight of $x$; define $\Flip{x}$ such that
$\Flip{x}_i = 1 - x_i, \forall i \in [n]$; denote the bitwise XOR of $x, y$ by
$x \oplus y$; define the corresponding $n$-qubit basis state as $\ket{x}$ and
$\ket{x_{\qty{i_1, \dots, i_k}}} = \ket{x_{i_1}, \dots, x_{i_k}}$.
Denote the complex conjugate of $\alpha \in \mathbb{C}$ by $\Conj{\alpha}$.

\paragraph{(Sparse) quantum state}
In a quantum system with $n$ qubits, the \emph{computational basis states} are
$2^n$ orthonormal vectors $\ket{0\dotsc 00}, \ket{0\dotsc 01}, \dotsc,
  \ket{1\dotsc 11}$, each identified as a binary string of length $n$.
A general \emph{quantum state} can be expressed as a normalized linear
combination of these computational basis states:
\begin{equation}
  \ket{\phi} = \sum_{x \in \BCube{n}} c_x \ket{x}, \quad \sum_{x \in
    \BCube{n}} \abs{c_x}^2 = 1.
\end{equation}
The coefficients $c_x$ are known as \emph{amplitudes}.
A quantum state is considered \emph{sparse} if only a small fraction of its
amplitudes are non-zero.
The \emph{sparsity} of a quantum state refers to the number of non-zero
amplitudes.
Call the set of binary strings corresponding to non-zero amplitudes the
\emph{binary support} of a sparse quantum state.

\paragraph{Quantum gate and quantum circuit}
A $k$-qubit \emph{quantum gate} is a $2^k\times 2^k$ unitary transformation
that acts on a $k$-qubit subsystem.
Our work utilizes the following quantum gates:
\begin{itemize}
  \item The $X$ gate is a 1-qubit gate that flips a qubit between $\ket{0}$ and
        $\ket{1}$.
  \item The CNOT gate is a 2-qubit gate that applies an $X$ gate to the \emph{target
          qubit} when the \emph{control qubit} is $\ket{1}$.
  \item The Toffoli gate is a 3-qubit gate that applies an $X$ gate to a target qubit
        when the other two control qubits are both $\ket{1}$.
        It can be decomposed into 10 1-qubit gates and 6 CNOT gates
        \cite{nielsen2010quantum}.
  \item The $C^t X$ is a $(t+1)$-qubit gate that applies an $X$ gate to a target qubit
        when the other $t$ target qubits are all $\ket{1}$.
        This gate is known as a multi-controlled Toffoli gate, and can be decomposed
        into $\Ord{n}$ 1-qubit gates and CNOT gates \cite{gidney2015constructing}.
  \item The $\Spl{\alpha}{\beta}$ gate is a 1-qubit gate with
        $\alpha\in\mathbb{C},\beta\in\mathbb{R}_{+}$ and $\abs{\alpha} \leq \beta$.
        It is defined as:
        \begin{equation}
          \Spl{\alpha}{\beta} = \frac{1}{\beta}
          \begin{pmatrix}
            -\sqrt{\beta^2 - \abs{\alpha}^2} & \alpha                          \\
            \Conj{\alpha}                    & \sqrt{\beta^2 - \abs{\alpha}^2}
          \end{pmatrix}
          .
        \end{equation}
  \item The $C^t \Spl{\alpha}{\beta}$ gate is a $(t+1)$-qubit gate that applies the
        $\Spl{\alpha}{\beta}$ gate to a target qubit when the other $t$ target qubits
        are all $\ket{1}$.
        This gate can be decomposed into two 1-qubit gates and one multi-controlled
        Toffoli gate \cite{barenco1995elementary}.
\end{itemize}
The qubits a gate acts on are indicated by writing their indices in the
subscript of the gate.
In the case of (multi-)controlled gates, the indices of the control qubits and
the target qubit are separated by a semicolon.
For example, $X_j$ represents the $X$ gate acting on the $j$th qubit, and
$\Tof_{i,j;k}$ represents the Toffoli gate with the $i$th and $j$th qubit as
controls and the $k$th qubit as the target.

An $n$-qubit \emph{quantum circuit} is a sequence of quantum gates that
implements a $2^n \times 2^n$ unitary transformation.
The set of all 1-qubit gates and CNOT gates is considered \emph{universal}
because they can be used to implement any unitary transformation.
For convenience, we refer to 1-qubit gates and CNOT gates as \emph{elementary
  gates}.
In our study, the \emph{size} of a circuit is determined by the total number of
elementary gates after decomposing all gates into elementary ones.

We will also encounter the concept of \emph{parameterized gates} or
\emph{circuit}.
A $k$-qubit parameterized gates are maps from some real or complex parameters
to $2^k \times 2^k$ unitaries.
For example, $\Spl{\alpha}{\beta}$ and $C^t \Spl{\alpha}{\beta}$ are said to be
parameterized gates, if $\alpha, \beta$ are unspecified variables.
Similarly, an $n$-qubit parameterized circuit is a sequence of
(un)parameterized gates, that implements a map from some parameters to $2^n
  \times 2^n$ unitaries.

\paragraph{Problem formulation}
In this work, we consider the preparation of sparse quantum states using
quantum circuit with quantum elementary gates.
Specifically, we are given a set of binary strings $S = \qty{x^j \in \BCube{n}
    : 1 \leq j \leq s}$ of length $n$ with a size of $s$, and amplitudes $\qty{c_j
    \in \mathbb{C} : 1 \leq j \leq s}$ satisfying $\sum_{j=1}^s \abs{c_j}^2 = 1$.
Our goal is to output a quantum circuit $C_{S,\bm{c}}$ that prepares the $n$-qubit,
$s$-sparse quantum state $\ket{\phi_{S,\bm{c}}} = \sum_{j=1}^s c_j \ket{x^j}$ using
$m$ ancillary qubits, such that:
\begin{equation}
  C_{S,\bm{c}}(\ket{0^n}\ket{0^m}) = \ket{\phi_{S,\bm{c}}}\ket{0^m}.
\end{equation}
Our objective is to minimize the size of the quantum circuit $C_{S,\bm{c}}$.

\section{Review of the CVO-QRAM algorithm}\label{sec:8aj1}
The CVO-QRAM algorithm \cite{de2022double} constructs a circuit of size
$\Tht{\sum_{j=1}^s \HW{x^j}}$ with one ancillary qubit to prepare any $n$-qubit
$s$-sparse state $\ket{\phi} = \sum_{j=1}^s c_j \ket{x^j}$.
In this section, we briefly review this algorithm.
We will show how to improve upon this algorithm in the following sections.

CVO-QRAM loads the amplitudes and binary strings one by one in a specific
order, utilizing one ancillary qubit to distinguish between loaded and being
processed terms.
Let $\Reg{\bm{m}}{M}$ denote the $n$-qubit \emph{memory register}, and
$\Reg{f}{F}$ denote the 1-qubit \emph{flag register}.
At the beginning of the algorithm, the registers are in the state
$\Reg{0^n}{M}\Reg{1}{F}$.
As proved in the following lemma, the algorithm deterministically transforms
the registers into the state $\Reg{\phi}{M}\Reg{0}{F}$.

\begin{lemma}[CVO-QRAM \cite{de2022double}]
  Any $n$-qubit $s$-sparse state $\ket{\phi} = \sum_{j=1}^s c_j \ket{x^j}$ can be
  prepared using $\Tht{\sum_{j=1}^s \HW{x^j}}$ elementary gates and one ancillary
  qubit.
\end{lemma}

\begin{proof}
  We may assume that $(c_1, x^1), \dots, (c_s, x^s)$ are sorted such that
  $\HW{s^j} \leq \HW{s^{j+1}}$ for $1 \leq j \leq s-1$.
  Let $\gamma_j := \sqrt{1 - \sum_{i=1}^{j-1} \abs{c_i}^2}$ for $1 \leq j \leq
    s$.
  \Cref{alg:4som} performs the transformation:
  \begin{equation}
    \Reg{0^n}{M}\Reg{1}{F} \to \Reg{\phi}{M}\Reg{0}{F}.
  \end{equation}

  \begin{algorithm}
    \raggedright

      \Desc{Input}{$\qty{(c_j, x^j)}_{1 \le j \le s}$}
      \Desc{Output}{A circuit performing $\Reg{0^n}{M}\Reg{1}{F} \to \Reg{\phi}{M}\Reg{0}{F}$}
      \vspace{1ex}

    For the $j$th string, with $1 \le j \le s$:
    \begin{enumerate}
      \item $\Reg{0^n}{M} \to
              \Reg{x^j}{M}$ when $\Reg{f}{F} = \ket{1}$.
      \item Apply $C^{t} \Spl{c_j}{\gamma_j}$ with controls on all $t = \HW{x^j}$ qubits
            where $x^j_k = 1$ and target on $\Reg{f}{F}$.
      \item $\Reg{x^j}{M} \to \Reg{0^n}{M}$ when $\Reg{f}{F} = \ket{1}$.
    \end{enumerate}

    \caption{\raggedright The CVO-QRAM algorithm.}
    \label{alg:4som}
  \end{algorithm}

  We first prove the correctness of \cref{alg:4som}.
  It can be verified  that at the beginning of the $j$th stage for $1 \leq j
    \leq s$, the registers are in the following state:
  \begin{equation}
    \sum_{i=1}^{j-1} c_{i}\Reg{x^{i}}{M}\Reg{0}{F} + \gamma_j\Reg{0^{n}}{M}\Reg{1}{F}.
  \end{equation}
  Hence, at the end of the algorithm, the registers will be in the state
  $\Reg{\phi}{M}\Reg{0}{A}$ with probability 1.

  Next, let us analyze the circuit size.
  In the $j$th stage, the 1st and 3rd step can each be implemented by $\HW{x^j}$
  CNOT gates, while the 2nd step can be implemented by $\Tht{\HW{x^j}}$
  elementary gates.
  Therefore, the entire circuit can be implemented by $\Tht{\sum_{j=1}^s
      \HW{x^j}}$ elementary gates.
\end{proof}

We make two remarks regarding the CVO-QRAM algorithm.
\begin{enumerate}
  \item The size of the constructed circuit is proportional to the number of 1s in all
        the $s$ binary strings.
        Hence, this algorithm excels when the state to prepare is \emph{double sparse},
        i.e., the state is sparse \emph{and} the binary strings in the
        binary support have low Hamming weights.
        In the worst case, the circuit size of this algorithm is still $\Tht{sn}$.
        In the next section, we will demonstrate that the total number of 1s
        can be reduced for \emph{any} set of binary strings, thus improving the worst
        case performance compared to CVO-QRAM.
  \item The $\Spl{c_j}{\gamma_j}$ gate has to be multi-controlled to ensure that the
        loaded terms are not affected.
        The straightforward solution is to apply the $\Spl{c_j}{\gamma_j}$ gate when
        $\Reg{\bm{m}}{M} = \ket{x^{j}}$.
        However, this implementation requires $\Tht{n}$ elementary gates.
        The trick employed by the CVO-QRAM algorithm is to sort the binary strings
        according to their Hamming weights, resulting in a
        reduction in the number of control qubits.
        The multi-controlled $\Spl{c_j}{\gamma_j}$ will also be the
        bottleneck of our algorithms, and we will use different techniques to address
        it.
\end{enumerate}

\section{BE-QRAM}\label{sec:yakv}
In this section, we present an enhanced version of the CVO-QRAM
algorithm, which we refer to Batch-Elimation QRAM (BE-QRAM), that
reduces the worst-case circuit size from
  $\Ord{sn}$ to $\Ord{sn / \log n}$.
Additionally, we establish a corresponding lower bound for a specific class of
algorithms.

As mentioned in the last section, the cost of the CVO-QRAM algorithm is propotional to the number of 1s in all the $s$ binary strings.
To achieve this improvement, we eliminate 1s in the binary strings being
prepared, leveraging the insight that any $\log_2(n) - \omg{1}$ binary strings
have at most $\ord{n}$ distinct positions, as supported by the following lemma.
It is worth noting that this observation has been previously utilized to
optimize reversible circuits \cite{markov2008optimal,zakablukov2017asymptotic}.

\begin{lemma}
  \label{lem:dgsw}
  Let $n, k, t, r \in \mathbb{N}$ where $k < \log n, t = 2^{k}, r = n - t$.
  Given $S_k = \qty{y^{1}, \dots, y^{k}} \subseteq \BCube{n}$, there exists an
  index set $\mathcal{T} \subseteq [n]$ of size $t$ such that for all $i \in \mathcal{R} = [n]
    \setminus \mathcal{T}$, there exists $j \in \mathcal{T}$ satisfying:
  \begin{equation}
    y_{i}^{p} = y_{j}^{p}, \quad \forall p \in [k].
  \end{equation}
  Furthermore, there exists a circuit $C$ consisting of $r$ CNOT gates that
  transforms $\ket{y^{p}}$ to $\ket{\Elm{y}^{p}}$ for all $p \in [k]$, where
  $\Elm{y}^{p} \in \{0,1\}^{n}$ and
  \begin{equation}
    \Elm{y}_{i}^{p} = 0, \quad \forall i \in \mathcal{R}.
  \end{equation}
\end{lemma}

\begin{proof}
  Let us call the length-$k$ binary string $y_{i}^{1} \dots y_{i}^{k}$ the
  \textit{pattern} of $S_{k}$ at position $i$.
  Since there are at most $t = 2^k$ distinct patterns, there exist $t$ patterns
  such that each of the remaining $r = n - t$ patterns coincides with one of
  them.
  We define $\mathcal{T}$ as the set of indices of these $t$ patterns.

  For every $i \in \mathcal{R}$, let $l(i)$ be an index in $\mathcal{T}$ such
  that the pattern at position $l(i)$ coincides with the pattern at position $i$.
  It can be verified that the following CNOT circuit accomplishes the desired
  transformation:
  \begin{equation}
    C = \prod_{i \in \mathcal{R}} \CNOT_{l(i);i}.
  \end{equation}
\end{proof}

As an example, \cref{fig:qdzp} illustrates the CNOT circuit that eliminates 1s
at the last 2 positions for 2 binary strings of length 6.

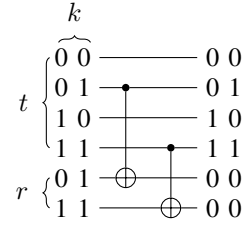
\begin{figure}
  \centering
  \begin{tikzpicture}
    \begin{yquant}

      [name=first]
      qubit {0~0} q;
      qubit {0~1} q[+1];
      qubit {1~0} q[+1];
      [name=mid 1]
      qubit {1~1} q[+1];
      [name=mid 2]
      qubit {0~1} q[+1];
      [name=last]
      qubit {1~1} q[+1];

      cnot q[4] | q[1];
      cnot q[5] | q[3];

      output {0~0} q[0];
      output {0~1} q[1];
      output {1~0} q[2];
      output {1~1} q[3];
      output {0~0} q[4];
      output {0~0} q[5];
    \end{yquant}

    \begin{scope}[decoration=brace]
      \draw [decorate] (mid 1.west) -- (first.west) node [midway,left=1ex] {$t$};
      \draw [decorate] (last.west) -- (mid 2.west) node [midway,left=1ex] {$r$};
      \draw [decorate,transform canvas={yshift=.5ex}] (first.150) -- (first.30) node [midway,above=1ex] {$k$};
    \end{scope}
  \end{tikzpicture}
  \caption{CNOT circuit eliminating 1s at the last 2 positions.}
  \label{fig:qdzp}
\end{figure}

We are now ready to present our first algorithm.
To begin with, we introduce parameters $k < \log_{2} n$, $t = 2^{k}$, and $r =
  n - t$ as in \cref{lem:dgsw}.
The specific values for these parameters will be determined later.
For simplicity, we assume that $k$ divides $s$ for now.

Our algorithm splits the $s$ binary strings into $\frac{s}{k}$ \emph{batches},
where each batch contains $k$ strings.
We will load these batches one by one, eliminating 1s in each batch following
the method described in \cref{lem:dgsw}.
Moreover, to reduce the cost of multi-controlled $\Spl{\alpha}{\beta}$ gates, we
introduce an additional ancillary qubit compared to the CVO-QRAM algorithm.
The detailed steps of the algorithm can be found in the proof of the following
theorem.

\begin{theorem}
  An $n$-qubit $s$-sparse state can be prepared by a circuit of size
  $\Ord{\frac{ns}{\log n} + n}$ with two ancillary qubits.
\end{theorem}

\begin{proof}
  Let $\ket{\phi} = \sum_{i=1}^{s/k} \sum_{j=1}^k c_{i,j} \ket{x^{i,j}}$ be the state
  to be prepared, and let
  \begin{equation}
    \gamma_{i,j} = \sqrt{1 - \sum_{i'=1}^{i-1}
      \sum_{j'=1}^{k} \abs{c_{i',j'}}^2 - \sum_{j'=1}^{j-1} \abs{c_{i,j'}}^2}.
  \end{equation}

  In addition to the memory register $\Reg{\bm{m}}{M}$ and the flag register
  $\Reg{f}{F}$, we introduce an extra 1-qubit ancillary register $\Reg{a}{A}$.
  \Cref{alg:ep62} performs the transformation
  \begin{equation}
    \Reg{0^n}{M}\Reg{0}{A}\Reg{1}{F} \to \Reg{\phi}{M}\Reg{0}{A}\Reg{0}{F}.
  \end{equation}

  \def\Cond{$\Reg{m_{\mathcal{T}}}{M} = \ket{\Elm{x}^{i,j}_{\mathcal{T}}}$ and $\Reg{a}{A} = \ket{1}$}

  \begin{algorithm}
    \raggedright

      \Desc{Input}{$\qty{(c_j, x^j)}_{1 \le j \le s}$}
      \Desc{Output}{A circuit performing $\Reg{0^n}{M}\Reg{0}{A}\Reg{1}{F} \to \Reg{\phi}{M}\Reg{0}{A}\Reg{0}{F}$}
      \vspace{1ex}

    For the $i$th batch, with $1 \leq i \leq \frac{s}{k}$:
    \begin{enumerate}
      \item \label{itm:vbpm}
            Apply the elimination circuit according to
            \cref{lem:dgsw}, eliminating all 1s at positions $\mathcal{R}$ for binary
            strings in the $i$th batch.
            Recall that $\abs{\mathcal{R}} = r, \mathcal{T} = [n] \setminus \mathcal{R}$
            and $\Elm{x}^{i,j}$ denotes the eliminated form of string $x^{i,j}$.
      \item \label{itm:lkie}
            $\Reg{0}{A} \to \Reg{1}{A}$ when $\Reg{m_{\mathcal{R}}}{M} = \ket{0^r}$.
      \item For the $j$th string in the $i$th batch, with $1 \leq j \leq k$:
            \begin{enumerate}
              \item \label{itm:prf2}
                    $\Reg{0^n}{M} \to \ket{\Elm{x}^{i,j}}$ when $\Reg{f}{F} = \ket{1}$.
              \item \label{itm:mph2}
                    Apply $\Spl{c_{i,j}}{\gamma_{i,j}}$ on $\Reg{f}{F}$ when \Cond.
              \item \label{itm:0fwe}
                    $\Reg{\Elm{x}^{i,j}}{M} \to \ket{0^n}$ when $\Reg{f}{F} = \ket{1}$.
            \end{enumerate}
      \item \label{itm:kca0}
            $\Reg{1}{A} \to \Reg{0}{A}$ when $\Reg{m_{\mathcal{R}}}{M} = \ket{0^r}$.
      \item \label{itm:mnld}
            Apply the reverse of the elimination circuit.
    \end{enumerate}

    \caption{\raggedright BE-QRAM}
    \label{alg:ep62}
  \end{algorithm}

  Let us first show the correctness of \cref{alg:ep62}.
  Observe the following two facts:
  \begin{itemize}
    \item $\Elm{x}^{i,j}$ are binary strings different from each other, since the
          elimination circuit is linear reversible.
    \item In step \ref{itm:mph2}, the conditions ``\Cond'' is equivalent
          to ``$\Reg{\bm{m}}{M} = \ket{\Elm{x}^{i,j}}$''.
  \end{itemize}
  From these observations, it can be verified that at the beginning of the $i$th
  batch, the registers are in the state:
  \begin{equation}
    \sum_{i'=1}^{i-1} \sum_{j'=1}^{k} c_{i',j'} \Reg{x^{i',j'}}{M}\Reg{0}{A}
    \Reg{0}{F} + \gamma_{i,1}\Reg{0^{n}}{M}\Reg{0}{A}\Reg{1}{F}.
  \end{equation}
  Hence, at the end of the algorithm, the registers will be in the state
  $\Reg{\phi}{M}\Reg{0}{A}\Reg{0}{F}$ with probability 1.

  Next, we analyze the size of the constructed circuit.
  \begin{itemize}
    \item Step \labelcref{itm:vbpm,itm:mnld} can each be implemented by $r$ CNOT gates,
          according to \cref{lem:dgsw}.
    \item Step \labelcref{itm:lkie,itm:kca0} can each be implemented by $2r$ $X$ gates
          and one $C^r X$ gate, hence $\Ord{r}$ elementary gates.
    \item Step \labelcref{itm:prf2,itm:0fwe} can each be implemented by
          $\HW{\Elm{x}^{i,j}} \leq t$ CNOT gates.
    \item Step \ref{itm:mph2} can be implemented by $2(t - \HW{\Elm{x}^{i,j}}) \leq 2t$
          $X$ gate and one $C^{t+1} \Spl{c_{i,j}}{\gamma_{i,j}}$, hence $\Ord{t}$ elementary
          gates.
  \end{itemize}
  In total, the circuit size is:
    \begin{equation}
      \label{eq:4398}
      \Ord{\ceil{\frac{s}{k}} (r + kt)}.
    \end{equation}
  
  By taking $k = \log n - \log\log n$, we get the desired size $O(\frac{ns}{\log
      n} + n)$.
\end{proof}

Next, we show that the $\Ord{\frac{ns}{\log n} + n}$ size cannot be
asymptotically reduced for a certain class of algorithms.

\begin{theorem}
  Suppose $s \leq 2^{\delta n}$ with $0 \leq \delta < 1$.
  If an algorithm $\mathcal{A}$ for preparing $n$-qubit $s$-sparse states
  satisfies the following conditions:
  \begin{enumerate}
    \item $\mathcal{A}$ uses at most $\ply{n}$ ancillary qubits.
    \item $\mathcal{A}$ is not \emph{amplitude-aware}, i.e., for any $S \subseteq
            \BCube{n}$ and $c: S \to \mathbb{C}$ describing the state, $\mathcal{A}$
          outputs a parameterized circuit $U_S$ that depends only on $S$, along with
          parameters $\bm{\theta}_{S,c}$.
    \item $U_S \neq U_{S'}$ if $S \neq S'$.
  \end{enumerate}
  Then $\mathcal{A}$ requires $\Omg{\frac{ns}{\log n} + n}$ elementary gates in
  the worst case.
\end{theorem}

\begin{proof}
  We may assume that the circuit output by $\mathcal{A}$ is composed of CNOT and
  $R_z, R_y$ gates, since single qubit gates admit ZYZ decomposition
  \cite{nielsen2010quantum}.
  Let $K$ be the maximum circuit size, and suppose $\mathcal{A}$ uses
  at most $n^C$ qubits, where $C$ is a constant independent of $n$.

  $\mathcal{A}$ can place a CNOT gate in $n^C(n^C - 1)$ different ways, and place
  $R_z, R_y$ in $n^C$ different ways each.
  Hence, the number of different $U_S$ can not exceed $\qty(n^C(n^C - 1) + 2
    n^C)^K$.
  On the other hand, the number of different $S$ is $\binom{n}{s}$.
  Therefore, by the assumption of $\mathcal{A}$, we have:
  \begin{align}
    \qty(n^C(n^C - 1) + 2 n^C)^K & \geq \binom{2^n}{s}        \\
    \Longrightarrow K            & = \Omg{\frac{ns}{\log n}}.
  \end{align}
  Since the state to prepare has $n$ qubits, we have $K = \Omg{n}$.
  Consequently, $K = \Omg{\frac{ns}{\log n} + n}$.
\end{proof}

Remark that the trivial lower bound of circuit size is $\Ord{s + n}$, where
$\Ord{s}$ arises from the dimension, and $\Ord{n}$ arises from the number of
qubits.
We are able to prove a non-trivial lower bound by imposing constraints to
circuit generating algorithms.
However, we believe that such constraints are reasonable, in a sense that they
are satisfied by previous (sparse) quantum state preparation algorithms, except
for the 1st constraint.
We will elaborate on this point in the discussion section.

\section{LT-QRAM}\label{sec:ssul}
In this section, we present another algorithm based on the CVO-QRAM algorithm, which we refer to as Lazy-Tree QRAM
  (LT-QRAM), that excels when the $s$ binary strings can be iterated in a
\emph{loopless} way.
That is, there exists a procedure that generate these strings in sequence, such that the first one is produced in linear time and each subsequent string in constant time.

To this end, we introduce the concept of Hamiltonian path for a set of binary
strings.

\begin{definition}
  Given a set $S \subseteq \BCube{n}$ of size $s$.
  Call an ordering of its elements $(x^1, \dots, x^s)$ a \emph{Hamiltonian path}
  of $S$, and define the \emph{length} of this path by $\sum_{j=2}^{s}
    \HD{x^{j-1}}{x^j}$.
\end{definition}

Our second algorithm outputs a circuit whose size is linear in the length of a
Hamiltonian path, multiplied by a factor of $\Ord{\log n}$.

\begin{theorem}
  \label{thm:lvsc}
  An $n$-qubit $s$-sparse state, whose binary support admits a
  Hamiltonian path of length $L$, can be prepared by a circuit of size $\Ord{L
      \log n + n}$ with $\Ord{n}$ ancillary qubits.
\end{theorem}

\begin{proof}
  Suppose $\ket{\phi} = \sum_{j=1}^s c_j \ket{x^j}$ is the state to be prepared,
  and $(x^1, \dots, x^s)$ is a Hamiltonian path of length $L$.
  Let $\gamma_j := \sqrt{1 - \sum_{i=1}^{j-1} \abs{c_i}^2}$ for $1 \leq j \leq
    s$.
  For simplicity, we assume $n$ to be a power of 2.
  One can easily extend the proof for general $n$ without changing the asymptotic
  order.

  In addition to the memory register $\Reg{\bm{m}}{M}$ and the flag register
  $\Reg{f}{F}$, we introduce an $n_t = 2n-1$-qubit tree register
  $\Reg{\bm{t}}{T}$ and an $n_h = \log_2 n + 1$-qubit ancillary register
  $\Reg{\bm{h}}{H}$.
  $\Reg{\bm{t}}{T}$ is structured as a full binary tree, whose leafs are in a
  one-to-one correspondence to the $n$ memory qubits $\Reg{\bm{m}}{M}$.
  Denote the leaf qubit corresponding to $\Reg{m_i}{M}$ by $\Reg{q^i_1}{T}$, and
  its father, grandfather, et al., by $\Reg{q^i_2}{T}, \dots,
    \Reg{q^i_{n_h}}{T}$.
  Hence, $\Reg{q^i_{n_h}}{T}$ is the root of the tree for all $1 \le i \le n$.
  Denote the sibling qubit of $\Reg{q^i_j}{T}$ by $\Reg{s^i_j}{T}$, for $1 \leq j
    < n_h$.
  Define the depth of a qubit in $\Reg{\bm{t}}{T}$ by the length of path from it
  to the root qubit, so that the $\Reg{q^i_j}{T}$ has depth $n_h - j + 1$.
  The layout of these registers and notations are illustrated in \cref{fig:rpni},
  for the special case of $n = 8$ and $i = 5$.
  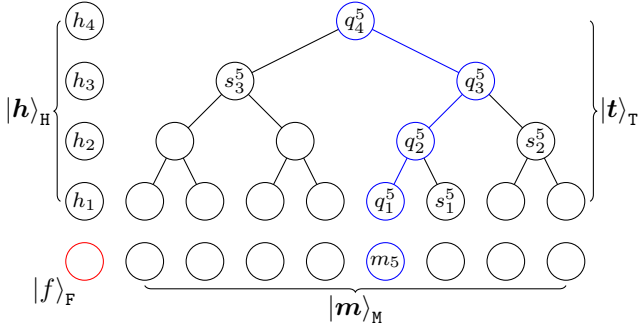
\begin{figure}
    \begin{tikzpicture}[fill=white,radius=0.25cm,x=0.8cm,y=0.8cm]
      \draw[draw=red] (0,0) circle;
      \draw ++(1,0) circle ++(1,0) circle ++(1,0) circle ++(1,0) circle ++(1,0) ++(1,0) circle ++(1,0) circle ++(1,0) circle;
      \draw ++(0,1) circle ++(0,1) circle ++(0,1) circle ++(0,1) circle;
      \draw (1,1) -- (1.5,2) (2,1) -- (1.5,2) (3,1) -- (3.5,2) (4,1) -- (3.5,2) (6,1) -- (5.5,2) (7,1) -- (7.5,2) (8,1) -- (7.5,2);
      \filldraw +(1,1) circle +(2,1) circle +(3,1) circle +(4,1) circle +(6,1) circle +(7,1) circle +(8,1) circle;
      \draw (1.5,2) -- (2.5,3) (3.5,2) -- (2.5,3) (7.5,2) -- (6.5,3);
      \filldraw +(1.5,2) circle +(3.5,2) circle +(7.5,2) circle;
      \draw (2.5,3) -- (4.5,4);
      \filldraw +(2.5,3) circle;
      \filldraw ;
      \draw[draw=blue] +(5,1) -- +(5.5,2) -- +(6.5,3) -- +(4.5,4);
      \filldraw[draw=blue] +(5,0) circle +(5,1) circle +(5.5,2) circle +(6.5,3) circle +(4.5,4) circle;
      \begin{scope}[decoration=brace,auto=left]
        \draw[decorate] (8,-0.4) -- +(-7,0) node[midway] {$\Reg{\bm{m}}{M}$};
        \draw[decorate] (8.4,4) -- +(0,-3) node[midway] {$\Reg{\bm{t}}{T}$};
        \draw[decorate] (-0.4,1) -- +(0,+3) node[midway] {$\Reg{\bm{h}}{H}$};
        \node at (-0.5,-0.5) {$\Reg{f}{F}$};
      \end{scope}
      \begin{scope}[font=\footnotesize]
        \node at (5,0) {$m_5$};
        \node at (5,1) {$q^5_1$};
        \node at (5.5,2) {$q^5_2$};
        \node at (6.5,3) {$q^5_3$};
        \node at (4.5,4) {$q^5_4$};
        \node at (6,1) {$s^5_1$};
        \node at (7.5,2) {$s^5_2$};
        \node at (2.5,3) {$s^5_3$};
        \node at (0,1) {$h_1$};
        \node at (0,2) {$h_2$};
        \node at (0,3) {$h_3$};
        \node at (0,4) {$h_4$};
      \end{scope}
    \end{tikzpicture}
    \caption{
      Layout of registers $\Reg{\bm{m}}{M}\Reg{\bm{t}}{T}\Reg{\bm{h}}{H}\Reg{f}{F}$
      in the proof of \cref{thm:lvsc}.
      The lines are meant for visualization, rather than hardware topology.
    }
    \label{fig:rpni}
  \end{figure}

  \Cref{alg:bnwj} performs the transformation
  \begin{equation}
    \Reg{0^n}{M}\Reg{0^{n_t}}{T}\Reg{0^{n_h}}{H}\Reg{1}{F} \to
    \Reg{\phi}{M}\Reg{0^{n_t}}{T}\Reg{0^{n_h}}{H}\Reg{0}{F}.
  \end{equation}

  \begin{algorithm}
    \raggedright

      \Desc{Input}{$\qty{(c_j, x^j)}_{1 \le j \le s}$}
      \Desc{Output}{A circuit performing $\Reg{0^n}{M}\Reg{0^{n_t}}{T}\Reg{0^{n_h}}{H}\Reg{1}{F} \to
          \Reg{\phi}{M}\Reg{0^{n_t}}{T}\Reg{0^{n_h}}{H}\Reg{0}{F}$}
      \vspace{1ex}

    \begin{enumerate}
      \item \label{itm:qhu7}
            $\Reg{0^n}{M}\Reg{0^{n_t}}{T} \to \Reg{x^1}{M}\Reg{1^{n_t}}{T}$.
      \item \label{itm:09d5}
            Apply $\Spl{c_1}{1}$ to $\Reg{f}{F}$.
      \item \label{itm:vfnu}
            For the $i$th string with $2 \leq i \leq s$:
            \begin{enumerate}
              \item \label{itm:o8uc}
                    $\Reg{x^{i-1}}{M} \to \Reg{x^i}{M}$ when $\Reg{f}{F} = \ket{1}$.
              \item \label{itm:rfio}
                    For each $j$ such that $x^{i-1}_j \neq x^i_j$:
                    \begin{enumerate}
                      \item \label{itm:atiu}
                            Apply $X_{h_1}$ when $\Reg{f}{F} = \ket{0}$.
                      \item \label{itm:n9bg}
                            For $2 \leq k \leq n_h$, apply $\Tof_{h_{k-1},s^j_{k-1};h_k}$.
                      \item \label{itm:hr2h}
                            For $1 \leq k \leq n_h$, apply $\CNOT_{h_k;q^j_k}$.
                      \item Reverse step \labelcref{itm:atiu,itm:n9bg}.
                    \end{enumerate}
              \item \label{itm:4mru}
                    Apply $\Spl{c_i}{\gamma_i}$ when the root qubit of $\Reg{\bm{t}}{T}$ is
                    $\ket{1}$.
            \end{enumerate}
      \item \label{itm:0fbr}
            For $1 \leq l < n_h$, and for each qubit of depth $l$th in the tree
            $\Reg{\bm{t}}{T}$, apply a Toffoli gate with controls on its two children and
            target on itself.
      \item \label{itm:vik7}
            For $1 \leq j \leq n$, apply $X^{1-x^s_j}_{m_j}\CNOT_{m_j;q^j_1}
              X^{1-x^s_j}_{m_j}$.
    \end{enumerate}

    \caption{\raggedright LT-QRAM}
    \label{alg:bnwj}
  \end{algorithm}

  Let us first show the correctness of \cref{alg:bnwj}.
  It suffices to show that
  \begin{enumerate}
    \item At the start of the iteration of step \ref{itm:vfnu}, the registers are in the
          state:
          \begin{multline}
            \sum_{i'=1}^{i-1} c_{i'} \Reg{x^{i'}}{M}\Reg{\bm{t}(x^{i'}, x^{i-1})}{T}
            \Reg{0^{n_h}}{H}\Reg{0}{F} \\
            + \gamma_i \Reg{x^{i-1}}{M}\Reg{1^{n_t}}{T}\Reg{0^{n_h}}{H}\Reg{1}{F}.
          \end{multline}
          Here, $\bm{t}(x, y)$ indicates that the tree register is an AND tree of whether
          segments of $x$ equals that of $y$.
          Formally, the leafs $\Reg{q^1_1 \dots q^n_1}{T} = \ket{\Flip{x \oplus y}}$, and
          each interior qubit takes the value of the AND of its two children.
          Hence, the root qubit is $\ket{1}$ if and only if $x = y$.
    \item Step \labelcref{itm:0fbr,itm:vik7} correctly uncompute the AND tree.
  \end{enumerate}
  For the first point, the correctness of step \labelcref{itm:qhu7,itm:09d5} is
  obvious.
  We need to show that step \ref{itm:rfio} correctly updates the AND tree, after
  that the correctness of \labelcref{itm:o8uc,itm:4mru} is obvious.
  By changing $x^{i-1}$ to $x^i$, only paths ending in the root and starting at
  leaves that correspond to the differing bits between $x^{i-1}$ and $x^i$, need
  to be updated.
  Moreover, these paths can be updated in sequence.
  Consider the path $\Reg{q^j_1}{T}, \dots, \Reg{q^j_{n_h}}{T}$.
  After step \labelcref{itm:atiu,itm:n9bg}, we have that when
  $\Reg{f}{F} = \ket{0}$,
  \begin{equation}
    \Reg{h_k}{H} =
    \begin{cases}
      \ket{1},                               & k = 1, \\
      \ket{s_1 \wedge \dots \wedge s_{k-1}}, & k > 1.
    \end{cases}
  \end{equation}
  Then after step \ref{itm:hr2h}, $\Reg{q^j_k}{T}$ will be updated to $\ket{q^j_k
      \oplus h_k}$.
  It is not hard to verify that this is the desired update.
  When $\Reg{f}{F} = \ket{1}$, $\Reg{\bm{t}}{T}\Reg{\bm{h}}{H}$ will not be
  changed.
  Hence, we have proven the first point.
  For the second point, first notice that before step \ref{itm:0fbr}, the
  registers are in the state:
  \begin{equation}
    \sum_{i=1}^{s} c_i \Reg{x^i}{M}\Reg{\bm{t}(x^i, x^s)}{T}
    \Reg{0^{n_h}}{H}\Reg{0}{F}.
  \end{equation}
  After step \ref{itm:0fbr}, all interior qubits will be set to $\ket{0}$, by the
  definition of the AND tree.
  Since the leafs $\Reg{q^1_1, \dots, q^n_1}{T} = \ket{\Flip{x^i \oplus x^s}}$
  for $x^i$, step \ref{itm:vik7} will set all leafs to $\ket{0}$ too.
  Hence, we have proven the second point.

  Next, we analyze the size of the constructed circuit.
  \begin{enumerate}
    \item Step \labelcref{itm:qhu7,itm:09d5} can be implemented by $\HW{x^1} + n_t + 1 =
            \Ord{n}$ 1-qubit gates.
    \item Step \ref{itm:o8uc} can be implemented by $\HD{x^{i-1}}{x^i}$ CNOT gates.
    \item Step \ref{itm:rfio} can be implemented by $\Ord{\HD{x^{i-1}}{x^i} n_h} =
            \Ord{\HD{x^{i-1}}{x^i} \log n}$ elementary gates.
    \item Step \ref{itm:4mru} can be implemented by $O(1)$ elementary gates.
    \item Step \labelcref{itm:0fbr,itm:vik7} can be implemented by $\Ord{n_t} = \Ord{n}$
          elementary gates.
  \end{enumerate}
  Recall the definition of $L$.
  In total, the circuit size is $\Ord{L \log n + n}$.
\end{proof}

Since the circuit size is linear in $L$, one would be tempted to find the
shortest Hamiltonian path given a set of binary strings.
While this task is in general computationally hard \cite{ernvall1985np} and the
length of the shortest Hamiltonian path could be $\Tht{sn}$
\cite{cohen1996traveling}, there are cases when there exists a short
Hamiltonian path \emph{and} the path is easy to compute.
  An example is the $U(1)$-invariant state \mbox{\cite{raveh2024deterministic}},
  including \emph{Bethe state} \mbox{\cite{van2021preparing}} and \emph{Dicke
    state} \mbox{\cite{bartschi2019deterministic}} as a special case, whose binary
  support consists of all binary strings of Hamming weight $k$.
As shown in the following corollary, in terms of circuit size, our algorithm
outperforms a recent work for this task \cite{raveh2024deterministic} by a
factor $\frac{n}{\log n}$, and is nearly asymptotically optimal.

\begin{corollary}
  Any $U(1)$-invariant state $\sum_{x \in \BCube{n}, \HW{x} = k} c_x \ket{x}$ can
  be prepared by a circuit of size $\Ord{\binom{n}{k} \log n}$ with $\Ord{n}$
  ancillary qubits.
  Moreover, the classical complexity for generating the circuit is also
  $\Ord{\binom{n}{k} \log n}$.
\end{corollary}

\begin{proof}
  For the $\binom{n}{k}$ binary string of Hamming weight $k$, there exists a
  Hamiltonian path of length $L = 2\binom{n}{k}$ and the path can be computed in
  $\Ord{L}$ classical time \cite{vajnovszki2006loop}.
  Hence, the corollary follows directly from \cref{thm:lvsc}.
\end{proof}

\section{Numerical results}
We numerically test and benchmark our algorithms and present the results in
\cref{fig:ph1x}.
Instead of circuit size, in the experiments we count the number of CNOT gates.
Remark that usually the CNOT count is illustrative of circuit size \cite{shende2004minimal}, as is the case in our experiments.
We benchmark our two algorithms separately, plotting the CNOT count against
qubit numbers up to 6000.

Our first algorithm, the BE-QRAM algorithm, undergoes testing on randomly
generated sparse states.
To sample an $n$-qubit $s$-sparse state, we uniformly draw $s$ binary strings
of length $n$ at random.
Since the amplitudes do not influence circuit size for both our algorithm and
previous ones, there is no need to sample amplitudes.
We benchmark our algorithm and the CVO-algorithm, computing the CNOT count normalized
by $(ns)^{-1}$ for varying $n$ and $s = n$.
As depicted in \cref{fig:g9rg}, the CNOT count of our algorithm
fits well as $\Tht{\frac{ns}{\log n - \log\log n}}$, while the
CVO-QRAM algorithm scales as $\Tht{ns}$.
Notice that after normalization, the cost of the CVO-QRAM algorithm remains
constant for increasing $n$, while ours decreases.

Our second algorithm, the LT-QRAM algorithm, is examined on $U(1)$-invariant
states.
The binary support of a $U(1)$-invariant state with $k$ down-spins in a chain
of length $n$ constitutes the set of all length-$n$ strings with a Hamming
weight of $k$.
Furthermore, there exists a minimal Hamiltonian path of length $2
  \binom{n}{k}$.
Consequently, the circuit size of our algorithm admits a closed form dependent
on $n$ and $k$.
We benchmark our algorithm and the previous algorithm \cite{raveh2024deterministic}, computing the CNOT count normalized
by $\qty(\binom{n}{k} k)^{-1}$ for varying $n$ and $k = \floor{\frac{n}{2}}$.
As illustrated in \cref{fig:tgda}, the CNOT count
of our algorithm fits well as $\Tht{\binom{n}{k} \log n}$, whereas
\cite{raveh2024deterministic} scales as
$\Tht{\binom{n}{k} k}$.
Notice that after normalization, the cost of \cite{raveh2024deterministic} remains
constant for increasing $n$, while ours decreases.

We implemented our algorithm in Python and ran the experiments on a 4.9 GHz Intel Core i7.

\begin{figure*}
  \centering
  \begin{subfigure}[b]{0.49\textwidth}
    \centering
    \includegraphics[width=\textwidth]{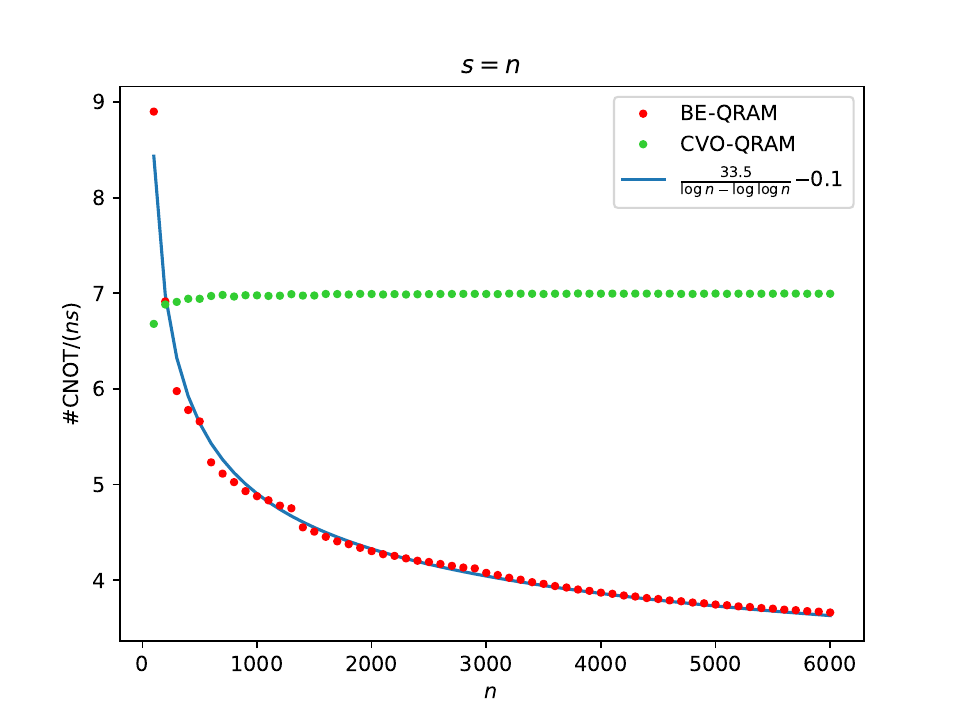}
    \caption{$s$-sparse state}
    \label{fig:g9rg}
  \end{subfigure}
  \hfill
  \begin{subfigure}[b]{0.49\textwidth}
    \centering
    \includegraphics[width=\textwidth]{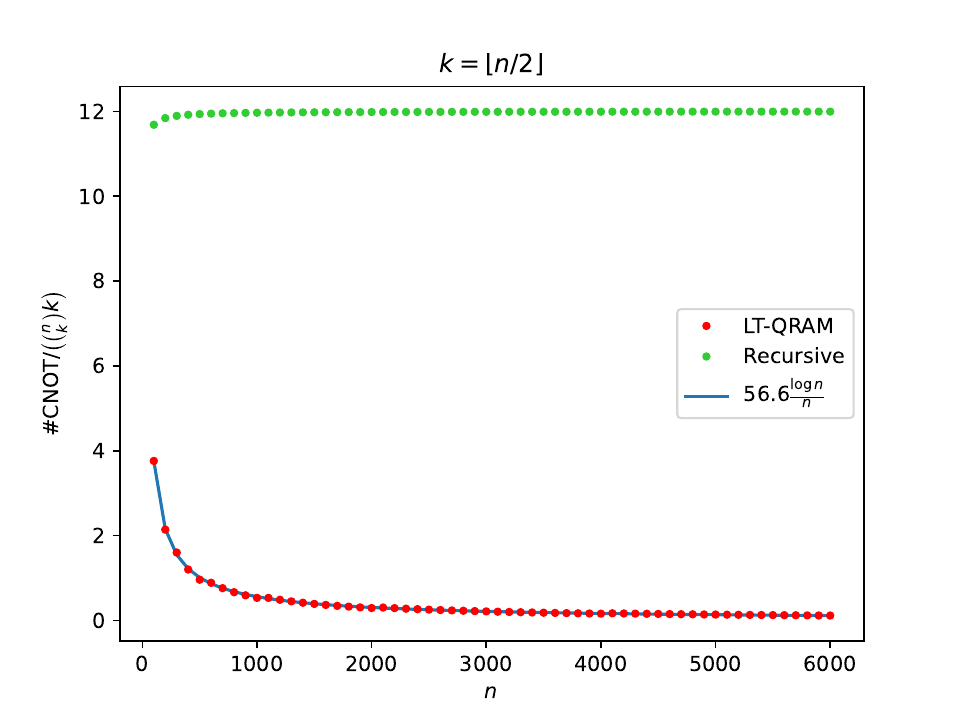}
    \caption{$U(1)$-invariant state with $k$ down-spins}
    \label{fig:tgda}
  \end{subfigure}
  \caption{Benchmark of our algorithms on sparse states and $U(1)$-invariant states. We plot the ``normalized'' number of CNOT gates against the number of qubits $n$ up to 6000 for (a) $s$ sparse state with $s = n$; (b) $U(1)$-invariant state with $k = \floor{n/2}$ down-spins. The red and green points represent the data of our algorithms and previous algorithms respectively, while the blue lines fit our data.}
  \label{fig:ph1x}
\end{figure*}

\section{Conclusion and discussion}\label{sec:u6jp}

In this work, we propose two efficient algorithms for the task of sparse
quantum state preparation, with a focus on optimizing circuit size.
The first algorithm achieves a size of $\Ord{\frac{ns}{\log n} + n}$, improving
previous upper bound by a factor of $\Ord{\log n}$.
A matching lower bound is shown for any amplitude-unaware algorithms with
moderate assumptions.
The second algorithm is tailored for sparse states whose binary
  support admits a short Hamiltonian path, achieving a size of $\Ord{L \log n}$,
where $L$ is the length of the path.
We show an application of the second algorithm to preparing $U(1)$-invariant
state, including Bethe states, outperforming previous work by a factor of
$\Ord{n / \log n}$.

We are only able to prove the optimality of the first algorithm for a certain
class of algorithms, rather than the problem itself.
It remains open how to close the $\Ord{s + n}$ vs $\Ord{\frac{ns}{\log n} + n}$
gap.
To the best of our knowledge, previous size lower bounds for quantum circuits
are obtained via the \emph{dimensionality argument}.
That is, the number of continuous variables, hence the circuit size, cannot be
fewer than the dimension of the problem.
In terms of sparse quantum state preparation, such argument only gives a lower
bound of $\Omg{s}$, which seems to contradict with the intuition that the
choice of bases $S$ has an impact on efficiency.
A possible bad case could be when the binary strings in $S$ differ from each
other at $\Omg{n}$ positions, which is possible whenever $s = 2^{\delta n}, 0
  \leq \delta < 1$ by the Gilbert–Varshamov bound \cite{cohen1996traveling}.
Based on our \emph{algorithmic} lower bound, there could be 3 possibilities:
\begin{enumerate}
  \item Either the true lower bound can be strengthened, which calls for new techniques
        for proving quantum circuit lower bounds.
  \item Or, there exists better algorithms, which we believe must be amplitude-aware.
  \item Or, there is a gap between the algorithmic and true lower bound, which means a
        compact circuit exists but is computationally hard to find.
\end{enumerate}
We tend to believe the first case to be true.

Another interesting open question is to decide whether our second algorithm
could be improved, especially for preparing $U(1)$-invariant states.
In this special case, our algorithm is a factor of $\Ord{\log n}$ away from the
lower bound obtained by dimensionality argument.
We believe this factor could be dropped.

\section*{Acknowledgment}
This work was supported in part by the National Natural Science Foundation of
China Grants No\period 62325210, 62272441, 12204489, 62301531, and the
Strategic Priority Research Program of Chinese Academy of Sciences Grant
No.XDB28000000.

\bibliographystyle{IEEEtran}

\begin{thebibliography}{10}
\providecommand{\url}[1]{#1}
\csname url@samestyle\endcsname
\providecommand{\newblock}{\relax}
\providecommand{\bibinfo}[2]{#2}
\providecommand{\BIBentrySTDinterwordspacing}{\spaceskip=0pt\relax}
\providecommand{\BIBentryALTinterwordstretchfactor}{4}
\providecommand{\BIBentryALTinterwordspacing}{\spaceskip=\fontdimen2\font plus
\BIBentryALTinterwordstretchfactor\fontdimen3\font minus \fontdimen4\font\relax}
\providecommand{\BIBforeignlanguage}[2]{{%
\expandafter\ifx\csname l@#1\endcsname\relax
\typeout{** WARNING: IEEEtran.bst: No hyphenation pattern has been}%
\typeout{** loaded for the language `#1'. Using the pattern for}%
\typeout{** the default language instead.}%
\else
\language=\csname l@#1\endcsname
\fi
#2}}
\providecommand{\BIBdecl}{\relax}
\BIBdecl

\bibitem{berry2015simulating}
D.~W. Berry, A.~M. Childs, R.~Cleve, R.~Kothari, and R.~D. Somma, ``Simulating hamiltonian dynamics with a truncated taylor series,'' \emph{Physical review letters}, vol. 114, no.~9, p. 090502, 2015.

\bibitem{low2017optimal}
G.~H. Low and I.~L. Chuang, ``Optimal hamiltonian simulation by quantum signal processing,'' \emph{Physical review letters}, vol. 118, no.~1, p. 010501, 2017.

\bibitem{low2019hamiltonian}
------, ``Hamiltonian simulation by qubitization,'' \emph{Quantum}, vol.~3, p. 163, 2019.

\bibitem{berry2015hamiltonian}
D.~W. Berry, A.~M. Childs, and R.~Kothari, ``Hamiltonian simulation with nearly optimal dependence on all parameters,'' in \emph{2015 {IEEE 56th} annual symposium on foundations of computer science}.\hskip 1em plus 0.5em minus 0.4em\relax IEEE, 2015, pp. 792--809.

\bibitem{schuld2015introduction}
M.~Schuld, I.~Sinayskiy, and F.~Petruccione, ``An introduction to quantum machine learning,'' \emph{Contemporary Physics}, vol.~56, no.~2, pp. 172--185, 2015.

\bibitem{biamonte2017quantum}
J.~Biamonte, P.~Wittek, N.~Pancotti, P.~Rebentrost, N.~Wiebe, and S.~Lloyd, ``Quantum machine learning,'' \emph{Nature}, vol. 549, no. 7671, pp. 195--202, 2017.

\bibitem{kerenidis2016quantum}
I.~Kerenidis and A.~Prakash, ``Quantum recommendation systems,'' \emph{arXiv preprint arXiv:1603.08675}, 2016.

\bibitem{rebentrost2018quantum}
P.~Rebentrost, A.~Steffens, I.~Marvian, and S.~Lloyd, ``Quantum singular-value decomposition of nonsparse low-rank matrices,'' \emph{Physical review A}, vol.~97, no.~1, p. 012327, 2018.

\bibitem{harrow2009quantum}
A.~W. Harrow, A.~Hassidim, and S.~Lloyd, ``Quantum algorithm for linear systems of equations,'' \emph{Physical review letters}, vol. 103, no.~15, p. 150502, 2009.

\bibitem{wossnig2018quantum}
L.~Wossnig, Z.~Zhao, and A.~Prakash, ``Quantum linear system algorithm for dense matrices,'' \emph{Physical review letters}, vol. 120, no.~5, p. 050502, 2018.

\bibitem{kerenidis2019q}
I.~Kerenidis, J.~Landman, A.~Luongo, and A.~Prakash, ``q-means: A quantum algorithm for unsupervised machine learning,'' \emph{Advances in neural information processing systems}, vol.~32, 2019.

\bibitem{kerenidis2021quantum}
I.~Kerenidis and J.~Landman, ``Quantum spectral clustering,'' \emph{Physical Review A}, vol. 103, no.~4, p. 042415, 2021.

\bibitem{rebentrost2014quantum}
P.~Rebentrost, M.~Mohseni, and S.~Lloyd, ``Quantum support vector machine for big data classification,'' \emph{Physical review letters}, vol. 113, no.~13, p. 130503, 2014.

\bibitem{plesch2011quantum}
M.~Plesch and {\v{C}}.~Brukner, ``Quantum-state preparation with universal gate decompositions,'' \emph{Physical Review A}, vol.~83, no.~3, p. 032302, 2011.

\bibitem{gleinig2021efficient}
N.~Gleinig and T.~Hoefler, ``An efficient algorithm for sparse quantum state preparation,'' in \emph{2021 58th {ACM/IEEE} Design Automation Conference {(DAC)}}.\hskip 1em plus 0.5em minus 0.4em\relax IEEE, 2021, pp. 433--438.

\bibitem{araujo2021entanglement}
I.~F. Araujo, C.~Blank, and A.~J. da~Silva, ``Entanglement as a complexity measure for quantum state preparation,'' \emph{J}, 2021.

\bibitem{rattew2022preparing}
A.~G. Rattew and B.~Koczor, ``Preparing arbitrary continuous functions in quantum registers with logarithmic complexity,'' \emph{arXiv preprint arXiv:2205.00519}, 2022.

\bibitem{dur2000three}
W.~D{\"u}r, G.~Vidal, and J.~I. Cirac, ``Three qubits can be entangled in two inequivalent ways,'' \emph{Physical Review A}, vol.~62, no.~6, p. 062314, 2000.

\bibitem{bartschi2019deterministic}
A.~B{\"a}rtschi and S.~Eidenbenz, ``Deterministic preparation of dicke states,'' in \emph{International Symposium on Fundamentals of Computation Theory}.\hskip 1em plus 0.5em minus 0.4em\relax Springer, 2019, pp. 126--139.

\bibitem{cottrell2019build}
W.~Cottrell, B.~Freivogel, D.~M. Hofman, and S.~F. Lokhande, ``How to build the thermofield double state,'' \emph{Journal of High Energy Physics}, vol. 2019, no.~2, pp. 1--43, 2019.

\bibitem{van2021preparing}
J.~S.~V. Dyke, G.~S. Barron, N.~J. Mayhall, E.~Barnes, and S.~E. Economou, ``Preparing bethe ansatz eigenstates on a quantum computer,'' \emph{PRX Quantum}, vol.~2, no.~4, p. 040329, 2021.

\bibitem{ben2005fast}
M.~Ben-Or and A.~Hassidim, ``Fast quantum byzantine agreement,'' in \emph{Proceedings of the thirty-seventh annual {ACM} symposium on Theory of computing}, 2005, pp. 481--485.

\bibitem{malvetti2021quantum}
E.~Malvetti, R.~Iten, and R.~Colbeck, ``Quantum circuits for sparse isometries,'' \emph{Quantum}, vol.~5, p. 412, 2021.

\bibitem{de2022double}
T.~M. {de Veras}, L.~D. da~Silva, and A.~J. da~Silva, ``Double sparse quantum state preparation,'' \emph{Quantum Information Processing}, vol.~21, no.~6, p. 204, 2022.

\bibitem{mozafari2022efficient}
F.~Mozafari, G.~D. Micheli, and Y.~Yang, ``Efficient deterministic preparation of quantum states using decision diagrams,'' \emph{Physical Review A}, vol. 106, no.~2, p. 022617, 2022.

\bibitem{ramacciotti2023simple}
D.~Ramacciotti, A.-I. Lefterovici, and A.~F. Rotundo, ``A simple quantum algorithm to efficiently prepare sparse states,'' \emph{arXiv preprint arXiv:2310.19309}, 2023.

\bibitem{grover2002creating}
L.~Grover and T.~Rudolph, ``Creating superpositions that correspond to efficiently integrable probability distributions,'' \emph{arXiv preprint quant-ph/0208112}, 2002.

\bibitem{zhang2022quantum}
X.-M. Zhang, T.~Li, and X.~Yuan, ``Quantum state preparation with optimal circuit depth: Implementations and applications,'' \emph{Physical Review Letters}, vol. 129, no.~23, p. 230504, 2022.

\bibitem{nielsen2010quantum}
M.~A. Nielsen and I.~L. Chuang, \emph{Quantum computation and quantum information}.\hskip 1em plus 0.5em minus 0.4em\relax Cambridge university press, 2010.

\bibitem{gidney2015constructing}
C.~Gidney, ``Constructing large controlled nots,'' \url{https://algassert.com/circuits/2015/06/05/Constructing-Large-Controlled-Nots.html}, 2015.

\bibitem{barenco1995elementary}
A.~Barenco, C.~H. Bennett, R.~Cleve, D.~P. DiVincenzo, N.~Margolus, P.~Shor, T.~Sleator, J.~A. Smolin, and H.~Weinfurter, ``Elementary gates for quantum computation,'' \emph{Physical review A}, vol.~52, no.~5, p. 3457, 1995.

\bibitem{markov2008optimal}
K.~Markov, I.~Patel, and J.~Hayes, ``Optimal synthesis of linear reversible circuits,'' \emph{Quantum Information and Computation}, vol.~8, no. 3\&4, pp. 0282--0294, 2008.

\bibitem{zakablukov2017asymptotic}
D.~V. Zakablukov, ``On asymptotic gate complexity and depth of reversible circuits without additional memory,'' \emph{Journal of Computer and System Sciences}, vol.~84, pp. 132--143, 2017.

\bibitem{ernvall1985np}
J.~Ernvall, J.~Katajainen, and M.~Penttonen, ``{NP-completeness} of the hamming salesman problem,'' \emph{BIT Numerical Mathematics}, vol.~25, pp. 289--292, 1985.

\bibitem{cohen1996traveling}
G.~Cohen, S.~Litsyn, and G.~Zemor, ``On the traveling salesman problem in binary hamming spaces,'' \emph{IEEE Transactions on Information Theory}, vol.~42, no.~4, pp. 1274--1276, 1996.

\bibitem{raveh2024deterministic}
D.~Raveh and R.~I. Nepomechie, ``Deterministic bethe state preparation,'' \emph{arXiv preprint arXiv:2403.03283}, 2024.

\bibitem{vajnovszki2006loop}
V.~Vajnovszki and T.~Walsh, ``A loop-free two-close gray-code algorithm for listing k-ary dyck words,'' \emph{Journal of Discrete Algorithms}, vol.~4, no.~4, pp. 633--648, 2006.

\end{thebibliography}


\end{document}